\documentclass[runningheads]{llncs}
\usepackage[T1]{fontenc}
\usepackage{graphicx}
\usepackage{url}
\usepackage{amsmath}
\usepackage{amssymb}
\usepackage{algorithm}
\usepackage{algpseudocode}
\usepackage{booktabs}
\usepackage{tabularx}
\usepackage{multirow}
\newcommand{\ef}{\mathrm{EF}}
\newcommand{\befmon}{B_{\ef}^{mon}}
\begin{document}
\title{Asymptotic Subsidies for Envy-Free Allocation}
\author{Zekai Wu\inst{1}\orcidID{0009-0007-1562-4827}}
\authorrunning{Zekai Wu}
\institute{The Hong Kong Polytechnic University, Kowloon, Hong Kong}
\maketitle              % typeset the header of the contribution
\begin{abstract}
We study the minimum monetary subsidy required for exact envy-freeness in allocating $m_n$ indivisible goods to $n$ agents. Valuations are additive, item values are i.i.d. from a distribution on $[0,1]$ with density bounded above and away from zero, and all agents value money equally. We consider $n\to\infty$ with $m_n=qn+r_n$, where $q\ge0$ is a fixed integer and $0\le r_n<n$. For the nonzero-remainder case where $1\le r_n<n$, we give a cubic-time algorithm based on a maximum-weight matching of balanced bundles and complementary dual prices. It returns an envy-free outcome for every instance. Under the random model, both its payment and the unrestricted minimum are $n-r_n+o_p(n)$ for $q\ge1$. When $q=0$ and $r_n/n$ converges to a limit below one, the error term improves to $O_p(1)$. For exact divisibility, where $r_n=0$, prior work gives zero subsidy with high probability for $q\ge2$. The square case $q=1$ is exceptional. Its minimum subsidy is $\Theta_p(\log n)$ and, with high probability, equals the optimum over allocations in which every agent receives exactly one good. This restricted optimum can be computed in cubic time using a maximum-weight perfect matching and all-pairs shortest paths. Together, these results provide a unified asymptotic characterization of the minimum subsidy.

\end{abstract}

\section{Introduction}

Fair allocation becomes challenging when shared resources cannot be divided arbitrarily among users~\cite{moulin2019internet,amanatidis2023survey}. In practical computing systems, for example, resources such as memory and accelerators often need to be allocated in discrete units~\cite{sivan2026quota,vuppalapati2023karma}. Such units are naturally modeled as indivisible goods. A central fairness criterion is \emph{envy-freeness} (EF)~\cite{foley1967resource,varian1974equity}, which requires every agent to value her own bundle at least as highly as any other agent's bundle. However, an EF allocation of the goods need not exist, even in the simple case of a single good and two agents. A natural approach is to supplement the allocation with compensation.

A related idea appears in computing systems, where credits are used to mediate resource allocation. Karma awards credits when a user's unused resource share is borrowed and uses the resulting balances to prioritize future access~\cite{vuppalapati2023karma}. Quota Marketplace uses internal credits and dynamic prices to allocate ML training resources across teams~\cite{sivan2026quota}. Although these credits serve system-specific operational purposes rather than act as monetary subsidies, they show how credit balances can accompany allocations of physical resources. We study a static fair-division analogue in which indivisible goods are allocated together with nonnegative monetary subsidies from an external budget. Money has the same unit value for every agent, and EF is evaluated over the combined allocation of goods and money.

The fair-division literature provides strong worst-case guarantees for this model. Halpern and Shah~\cite{halpern2019subsidy} characterized which allocations can be made envy-free through payments and studied the subsidies required. For $n$ agents with additive valuations in which each good has value at most one, Brustle et al.~\cite{brustle2020dollar} proved that a total subsidy of at most $n-1$ always suffices when the allocation can also be chosen. This bound is tight in the worst case, but does not describe the subsidy typically needed for a given number of goods. Indeed, Halpern and Shah's experiments on synthetic and Spliddit instances found that small subsidies often suffice~\cite{halpern2019subsidy}, motivating a theoretical understanding of the compensation required beyond the worst case.

The asymptotic fair-division literature offers a framework for this question by drawing agents' values for goods independently from a common distribution. Dickerson et al.~\cite{dickerson2014computational} initiated the study of envy-free allocations in this setting. Manurangsi and Suksompong~\cite{manurangsi2020envy,manurangsi2021closing} established existence guarantees when goods are sufficiently plentiful and showed that divisibility of the number of goods by the number of agents plays an important role. These existence results identify settings in which no compensation is needed. When an envy-free allocation of the goods alone does not exist, however, they do not determine the amount of compensation needed to restore envy-freeness.

We connect these two perspectives by asking:
\begin{quote}
\emph{How much total monetary compensation is needed for exact envy-freeness with high probability, and how does this amount depend on the number of goods relative to the number of agents?}
\end{quote}
We consider additive valuations with individual values drawn independently from a common distribution on $[0,1]$ whose density is bounded above and away from zero. For each realized instance, we minimize the total nonnegative compensation over all envy-free allocations of goods and money. A budget suffices with high probability if, with probability $1-o(1)$ as the number of agents grows, the instance admits an outcome within that budget satisfying every pairwise EF inequality exactly.

\paragraph{Our contributions.}
We determine the asymptotic behavior of the minimum total monetary compensation needed for envy-freeness. Our asymptotic results concern $n$ agents and $m_n=qn+r_n$ goods, where $q\ge0$ is a fixed integer and $0\le r_n<n$. Both the quotient $q$ and the remainder $r_n$ matter. For nonzero remainders, we use one balanced-bundle construction based on a maximum-weight matching and dual-price compensation. It is EF for every instance, runs in $O(n^3)$ time, and its total payment satisfies the corresponding estimate in Theorem~\ref{thm:monetary-nonzero-remainders}. The square case $m_n=n$ has a different, logarithmic scale. We prove that the unrestricted optimum coincides w.h.p. with the optimum over singleton allocations, in which each agent receives exactly one good. We compute this restricted optimum by matching and shortest paths in $O(n^3)$ time and derive a distribution-free root/choice-share representation. Together with the known zero-compensation result at exact divisibility for $q\ge2$, these results give the summary in Table~\ref{tab:monetary-summary}.

\begin{table}[htbp]
\caption{Summary of results for the minimum monetary compensation, where
$m_n=qn+r_n$ and $q$ is fixed. For $q=0$ and $r_n>0$, we assume that
$r_n/n\to\lambda\in[0,1)$.}
\label{tab:monetary-summary}
\centering
\small
\renewcommand{\arraystretch}{1.4}
\renewcommand{\tabularxcolumn}[1]{m{#1}}
\newcommand{\summarycell}[2]{%
\shortstack[c]{\rule{0pt}{2.7ex}#1\\[1pt]
{\footnotesize #2}\rule[-1.3ex]{0pt}{2.3ex}}}
\begin{tabularx}{\textwidth}{|
    >{\centering\arraybackslash}m{0.12\textwidth}|
    >{\centering\arraybackslash}X|
    >{\centering\arraybackslash}X|}
\hline
$q$ & $1\le r_n<n$ & $r_n=0$ \\
\hline
$0$
&
\summarycell{$n-r_n+O_p(1)$}{(Thm.~\ref{thm:monetary-nonzero-remainders})}
&
\summarycell{$0$}{(Trivial)}
\\
\hline
$1$
&
\multirow{2}{=}[-7pt]{\centering
\summarycell{$n-r_n+o_p(n)$}{(Thm.~\ref{thm:monetary-nonzero-remainders})}}
&
\summarycell{$\Theta_p(\log n)$}{(Thm.~\ref{thm:square-logarithmic-scale})}
\\
\cline{1-1}\cline{3-3}
$\ge2$
&
&
\summarycell{$0$ w.h.p.}{(Manurangsi and
Suksompong~\cite{manurangsi2020envy})}
\\
\hline
\end{tabularx}
\end{table}

\subsection{Related Work}

Our work connects two lines of research: asymptotic fair division and fair division with subsidies.

\paragraph{Asymptotic fair division.}
This line of work studies the existence and computation of fair allocations under random valuations. Dickerson et al.~\cite{dickerson2014computational} showed that, under suitable distributional assumptions, allocating each good to an agent who values it most yields an EF allocation with high probability when $m\ge Cn\log n$ for a sufficiently large constant $C$. Under the i.i.d. model considered here, Manurangsi and Suksompong~\cite{manurangsi2020envy} established the importance of divisibility, proving high-probability existence when $n\mid m$ and $m\ge2n$, while showing that nonexistence can persist for certain nondivisible ratios $m/n$. Their subsequent work~\cite{manurangsi2021closing} showed that round robin produces an EF allocation with high probability when $m\ge Cn\log n/\log\log n$, without a divisibility requirement. Further work allows distributions to differ across agents~\cite{bai2022asymmetric} or across goods~\cite{kale2025nonidentical}. Bai et al.~\cite{bai2022smoothed} studied smoothed utilities, obtained by randomly perturbing arbitrary additive base valuations, while Benad\`e et al.~\cite{benade2024beyond} established existence guarantees for a class of non-additive valuations under independent random relabelings of the goods. The probabilistic perspective has also been applied to proportionality~\cite{suksompong2016proportional}, allocation to groups~\cite{manurangsi2017groups}, and unequal entitlements~\cite{manurangsi2025asymptotic}. For chores, Manurangsi and Suksompong~\cite{manurangsi2025chores} proved that $m\ge2n$ suffices for high-probability EF existence under analogous distributional assumptions on disutilities, without requiring $n\mid m$. These results concern fairness without monetary compensation. We study the minimum total subsidy needed for exact EF under random valuations, including settings in which an EF allocation of the goods alone is unlikely to exist.

\paragraph{Fair division with subsidies.}
Halpern and Shah~\cite{halpern2019subsidy} showed that an allocation admits envy-eliminating subsidies if and only if it maximizes social welfare over all reassignments of its bundles. They also gave a polynomial-time method for computing the minimum subsidies for a fixed envy-freeable allocation using longest paths in its envy graph. For additive valuations with individual good values at most one, Brustle et al.~\cite{brustle2020dollar} established the tight worst-case bound of $n-1$ on the total subsidy. Lu et al.~\cite{lu2026optimal} extended this bound to mixed goods and chores with additive item values in $[-1,1]$. Caragiannis and Ioannidis~\cite{caragiannis2022limited} studied the computational problem of minimizing total subsidies, giving additive approximation guarantees for a fixed number of agents and hardness of approximation when the number of agents is part of the input. Choo et al.~\cite{choo2024house} considered minimum subsidies in house allocation, where each agent receives exactly one house and excess houses may remain unassigned. They proved NP-hardness in general and solved the case $m=n$ in polynomial time by finding a maximum-weight perfect matching and computing the minimum subsidies for it. Dai et al.~\cite{dai2026complexity} strengthened the hardness result to binary utilities. Our results quantify the minimum total subsidy under i.i.d. valuations. In particular, our square-case analysis shows that the unrestricted optimum coincides with the singleton optimum with high probability and determines its logarithmic scale.

\section{Preliminaries}

Let $[k]:=\{1,2,\ldots,k\}$. We study the problem of allocating a set of indivisible goods $M=[m]$ to a set of agents $N=[n]$ with additive valuations. Each agent $i\in N$ has a valuation function $v_i:2^M\to \mathbb{R}_{\ge 0}$ such that for any subset of goods $S\subseteq M$, $v_i(S)=\sum_{j\in S}v_i(\{j\})$. For a singleton good $j\in M$, we write $v_i(j)$ as shorthand for $v_i(\{j\})$. We assume that $v_i(j)\in[0,1]$ for each agent $i\in N$ and good $j\in M$.\footnote{This assumption is without loss of generality. For a common distribution supported on $[0, L]$ for a fixed $L>0$, we can scale the valuations by a constant factor $1/L$ to ensure that the maximum valuation for a single good of any agent is at most 1.}   An \emph{allocation} is a partition of the goods into $n$ \emph{bundles} $A=(A_1,\ldots,A_n)$, where $A_i$ is the bundle allocated to agent $i$. Denote by $\mathbb{A}$ the set of all possible allocations. We say that an allocation satisfies \emph{envy-freeness} (EF) if for every pair of agents $i,j\in N$, we have $v_i(A_i)\ge v_i(A_j)$.

For asymptotic analysis, the valuations $v_i(j)$ for any agent $i\in N$ and any good $j\in M$ are drawn independently from a given distribution $\mathcal{D}$ over $[0,1]$.  We assume that $\mathcal{D}$ is absolutely continuous on $[0,1]$ with probability density function $f$.  Moreover, there exist constants $0<\alpha\le\beta<\infty$ such that $\alpha \le f(x) \le \beta$ for almost every $x\in[0,1]$.  We refer to this condition as \emph{PDF-boundedness}, which we assume throughout the paper. It is worth noting that our algorithms take only the realized valuations as input and that the density bounds are used in the probabilistic analysis.

All limits are taken as $n\to\infty$. An event holds \emph{with high probability} (w.h.p.) if its probability tends to one. For random variables $Y_n$ and a deterministic sequence $a_n>0$, we write $Y_n=O_p(a_n)$ if, for every $\epsilon>0$, there is a constant $C_\epsilon<\infty$ such that $\Pr(|Y_n|>C_\epsilon a_n)\le\epsilon$ for all sufficiently large $n$, and $Y_n=o_p(a_n)$ if $Y_n/a_n$ converges to zero in probability. For nonnegative $Y_n$, we write $Y_n=\Theta_p(a_n)$ if, for every $\epsilon>0$, there are constants $0<c_\epsilon<C_\epsilon<\infty$ such that $\Pr(c_\epsilon a_n\le Y_n\le C_\epsilon a_n)\ge1-\epsilon$ for all sufficiently large $n$. The number of goods may depend on $n$, and we set $m=m_n=qn+r_n$, where $q\ge0$ is a fixed integer and $0\le r_n<n$.

To restore envy-freeness, we allow nonnegative monetary subsidies from an external budget. We normalize every agent's value for one unit of money to one. Let $X=(x_1,\ldots,x_n)$ be the subsidy vector, where $x_i\ge0$ is the amount given to agent $i$, and let $B:=\sum_{i=1}^n x_i$ denote the total subsidy. An outcome $(A,X)$ is \emph{envy-free} if $v_i(A_i)+x_i\ge v_i(A_j)+x_j$ for every pair of agents $i,j\in N$. Define the minimum total compensation by
\begin{align*}
    \befmon(n,m_n):=\min_{\substack{A\in\mathbb{A},\;X\in\mathbb{R}_{\ge0}^n\\(A,X)\text{ is EF}}}\sum_{i=1}^n x_i.
\end{align*}
When there is no ambiguity, we abbreviate this quantity by $\befmon$.

Due to the page limit, full proofs in the following sections are deferred to Appendices A and B.

\section{Nonzero remainders}
\label{sec:monetary-nonzero-remainders}

When $m_n$ is not divisible by $n$, a balanced partition has $r_n$ bundles of size $q+1$ and $n-r_n$ bundles of size $q$. This section shows that this one-good imbalance determines the leading subsidy. Our main result is the following theorem.

\begin{theorem}\label{thm:monetary-nonzero-remainders}
Fix an integer $q\ge0$ and suppose that $m_n=qn+r_n$, where $1\le r_n<n$. If $q=0$, assume additionally that $r_n/n\to\lambda<1$. Then
\[
    \befmon=
    \begin{cases}
        n-r_n+O_p(1), & q=0,\\
        n-r_n+o_p(n), & q\ge1.
    \end{cases}
\]
Moreover, Algorithm~\ref{alg:balanced-compensation} returns an EF outcome for every valuation profile in $O(n^3)$ time, and under the random model its total payment obeys the corresponding estimate above.
\end{theorem}

When $1\le r_n<n$, the same deterministic construction underlies both the sparse case $q=0$ and the case $q\ge1$. Fix a valuation-independent partition $P_1,\ldots,P_n$ of $M$ with $n-r_n$ bundles of size $q$ and $r_n$ bundles of size $q+1$. We include empty bundles when $q=0$ and use the convention $v_i(\varnothing)=0$. Set $w_{ij}:=v_i(P_j)$ for every $i,j\in N$. The maximum-weight matching primal is
\begin{align}
    \max\quad &\sum_{i\in N}\sum_{j\in N}w_{ij}y_{ij}
                                                        \label{eq:balanced-primal}\\
    \text{s.t.}\quad
    &\sum_{j\in N}y_{ij}\le1
        &&\text{for every $i\in N$},\nonumber\\
    &\sum_{i\in N}y_{ij}\le1
        &&\text{for every $j\in N$},\nonumber\\
    &y_{ij}\ge0
        &&\text{for every $i,j\in N$}.\nonumber
\end{align}
Its dual is
\begin{align}
    \min\quad &\sum_{i\in N}\phi_i+\sum_{j\in N}\psi_j
                                                        \label{eq:balanced-dual}\\
    \text{s.t.}\quad
    &\phi_i+\psi_j\ge v_i(P_j)
        &&\text{for every $i,j\in N$},\nonumber\\
    &\phi_i\ge0
        &&\text{for every $i\in N$},\nonumber\\
    &\psi_j\ge0
        &&\text{for every $j\in N$}.\nonumber
\end{align}
Although~\eqref{eq:balanced-primal} permits partial matchings, nonnegative weights on the complete balanced bipartite graph allow any integral optimum to be extended to a perfect matching without decreasing its weight.

\begin{algorithm}[t]
\caption{EF allocation and compensation via dual prices}
\label{alg:balanced-compensation}
\begin{algorithmic}[1]
    \Require Valuation $(v_i(g))_{i\in N,g\in M}$ and
    $m_n=qn+r_n$
    \Ensure An allocation $A$ and a nonnegative payment vector $X$ such that
    $(A,X)$ is EF
    \State By a fixed valuation-independent rule, partition $M$ into
    $P_1,\ldots,P_n$, with $n-r_n$ bundles of size $q$ and $r_n$ bundles of
    size $q+1$
    \State Set $w_{ij}\gets v_i(P_j)$ for every $i,j\in N$
    \State Compute a maximum-weight perfect matching under weights $w_{ij}$,
    together with a complementary optimum
    $(\phi,\psi)$ of~\eqref{eq:balanced-dual}
    \State For the bundle $P_j$ matched to agent $i$, set
    $A_i\gets P_j$ and $\pi_i\gets\psi_j$
    \State Set $H\gets\max_{i\in N}\pi_i$ and
    $x_i\gets H-\pi_i$ for every $i\in N$
    \State \Return $A$ and $X$
\end{algorithmic}
\end{algorithm}

The construction satisfies the following basic properties.

\begin{lemma}
\label{lem:balanced-compensation}
For every valuation profile, Algorithm~\ref{alg:balanced-compensation} returns an EF outcome, and its total payment satisfies
\begin{align}
    \sum_{i\in N}x_i
    =n\max_{j\in N}\psi_j-\sum_{j\in N}\psi_j.       \label{eq:balanced-payment}
\end{align}
The algorithm runs in $O(nm_n+n^3)$ arithmetic operations, and hence in $O(n^3)$ operations when $q$ is fixed.
\end{lemma}

\subsection{The sparse case: $q=0$}

Here $r_n=m_n$. The common construction assigns every singleton bundle to a distinct agent and matches each remaining agent to an empty bundle. The $q=0$ assertion of Theorem~\ref{thm:monetary-nonzero-remainders} gives $\befmon/n\rightarrow_p1-\lambda$. Thus, for any fixed $\epsilon\in(0,1-\lambda)$, a total subsidy of $(1-\lambda-\epsilon)n$ is insufficient to achieve envy-freeness w.h.p., whereas $(1-\lambda+\epsilon)n$ is sufficient w.h.p.

The remaining analysis separates two probabilistic estimates. The matching chosen by Algorithm~\ref{alg:balanced-compensation} determines a set of empty-bundle agents from the entire valuation matrix, so the upper-bound error must be controlled uniformly over all possible such sets. The lower bound instead depends on an order statistic of the agents' largest values. We isolate these two estimates first.

\begin{lemma}
\label{lem:sparse-uniform-deficiency}
Suppose that $0<m_n<n$ and $m_n/n\to\lambda<1$, and define
\[
    \Delta_n:=
    \max_{\substack{U\subseteq N\\|U|=n-m_n}}
    \sum_{g\in M}\left(1-\max_{i\in U}v_i(g)\right).
\]
Then $\Delta_n=O_p(1)$.
\end{lemma}

The next estimate concerns the lower bound. Among any $m_n+1$ agents, at least one must receive an empty bundle; we will apply this observation to the agents whose favorite goods have values closest to one.

\begin{lemma}
\label{lem:sparse-critical-row}
Suppose that $0<m_n<n$ and $m_n/n\to\lambda<1$. For every agent $i$, define $\eta_i:=1-\max_{g\in M}v_i(g)$, and let $\eta_{(m_n+1)}$ be the $(m_n+1)$-st smallest of these values. Then
\[
    (n-m_n)\eta_{(m_n+1)}=O_p(1).
\]
\end{lemma}

The preceding lemmas contain all probabilistic input for the sparse case.

\paragraph{Proof overview for $q=0$.}
For the upper bound, complementary slackness gives zero dual potentials for empty bundles and their matched agents in Algorithm~\ref{alg:balanced-compensation}; Lemma~\ref{lem:sparse-uniform-deficiency} then controls the construction's total payment uniformly over the possible empty-recipient set. For the lower bound, all empty-bundle agents must receive the same payment, and among the agents identified by Lemma~\ref{lem:sparse-critical-row}, at least one must be empty. This forces every EF outcome to incur the same leading payment as the construction, up to a tight random constant-order error.

\begin{remark}
The $O_p(1)$ error in the $q=0$ case of Theorem~\ref{thm:monetary-nonzero-remainders} is sharp. To see this, consider a single good and $n$ agents with valuations drawn from the uniform distribution on $[0,1]$. The optimal allocation assigns the good to the highest-valuing agent. All remaining $n-1$ agents receive empty bundles and must therefore receive the same payment; otherwise, an empty-bundle agent receiving less money would envy another empty-bundle agent. In an optimal solution, the receiver gets no payment, while the common payment to every other agent equals the second-highest value of the good. Let $v_{(n-1)}$ denote the second-largest of the $n$ valuations. Then $\befmon=(n-1)v_{(n-1)}$, and $(n-1)-\befmon=(n-1)(1-v_{(n-1)})$ converges in distribution to a $\mathrm\Gamma(2,1)$ random variable.
\end{remark}

\subsection{The case $q\ge1$}

We assume $q\ge1$ throughout this subsection. The next three lemmas also apply when $r_n=0$, as needed in Section~\ref{sec:monetary-exact-divisibility}.

Every balanced bundle is nonempty. The preceding argument for obtaining zero dual potentials from matched empty bundles therefore no longer applies. For the upper bound, we show that the dual bundle potentials, after accounting for bundle sizes, are nearly constant. For the lower bound, we distinguish two cases according to the total excess above bundle size $q$. Many excess goods provide many disjoint $(q+1)$-tuples, whereas few excess goods first produce one large payment that can then be propagated through bundles of size $q$. We begin with the dual estimate that controls the payment of Algorithm~\ref{alg:balanced-compensation}.

\begin{lemma}
\label{lem:positive-dual-oscillation}
Fix a deterministic partition $P_1,\ldots,P_n$ in which every bundle has size $q$ or $q+1$. For any maximum-weight perfect matching between agents and these bundles, write $i_j$ for the agent matched to $P_j$, and let $(\phi,\psi)$ be any complementary optimal solution to the corresponding assignment dual~\eqref{eq:balanced-dual}. Define $\theta_j:=\psi_j-|P_j|$. There exists a constant $C>0$, depending only on $q$ and $\alpha$, such that, w.h.p.,
\[
    \max_j\theta_j-\min_j\theta_j
    \le C n^{-1/(q+1)}(\log n)^{1+1/(q+1)}=o(1).
\]
This event holds simultaneously for every choice of the maximum-weight matching and complementary optimal dual solution.
\end{lemma}

For the lower bound, the allocation may depend on the entire valuation instance. We therefore need estimates that hold uniformly over all relevant families of goods, rather than estimates for a family fixed independently of the valuations. The first such estimate applies when we can extract many disjoint $(q+1)$-tuples from bundles containing at least $q+1$ goods.

\begin{lemma}
\label{lem:positive-large-coverage}
Let $h_n:=\lceil n^{2/3}\rceil$. There exists a constant $\gamma_1>0$, depending only on $q$ and $\alpha$, such that, for $\eta_n:=\gamma_1n^{-1/(3(q+1))}$, the following holds w.h.p. For every family $\mathcal T$ of at least $\lceil h_n/(q+1)\rceil$ pairwise disjoint $(q+1)$-subsets of goods, all but at most $h_n$ agents value every good in at least one member of $\mathcal T$ at least $1-\eta_n$.
\end{lemma}

When the total excess above bundle size $q$ is small, we use a single high $(q+1)$-tuple to create a \emph{seed} payment close to one. High $q$-subsets then transmit this payment first to many agents and subsequently to all but a negligible number of agents. The next lemma records the three uniform estimates supporting these two propagation steps.

\begin{lemma}
\label{lem:positive-small-propagation}
Set $h_n:=\lceil n^{2/3}\rceil$. There exist positive constants $\gamma_2$, $\gamma_3$, and $\kappa_0$, depending only on $q$ and $\alpha$. With
\[
    \eta_n':=\gamma_2n^{-1/(3(q+1))},\qquad
    \zeta_n:=\gamma_3n^{-1/(6q)},\qquad
    \ell_n:=\lfloor\kappa_0n^{5/6}\rfloor,
\]
the following properties hold simultaneously w.h.p.:
\begin{enumerate}
    \item every $(q+1)$-subset of goods has at least $4h_n$ agents who value each of its goods at least $1-\eta_n'$;
    \item every $q$-subset of goods has at least $\ell_n$ agents who value each of its goods at least $1-\zeta_n/q$;
    \item for every family of at least $\lceil\ell_n/2\rceil$ pairwise disjoint $q$-subsets, all but at most $\lceil n^{1/3}\rceil$ agents value every good in at least one member of the family at least $1-\zeta_n/q$.
\end{enumerate}
\end{lemma}

Together, the three lemmas above provide all probabilistic estimates for the case $q\ge1$.

\paragraph{Proof overview for $q\ge1$.}
For the upper bound, Lemma~\ref{lem:positive-dual-oscillation} makes the adjusted bundle potentials nearly constant, so the payment identity in Lemma~\ref{lem:balanced-compensation} is governed by the number of smaller bundles. For the lower bound, consider any EF allocation and split according to the total excess above bundle size $q$. If the excess is large, extract many disjoint $(q+1)$-tuples and apply Lemma~\ref{lem:positive-large-coverage} to force payments broadly. If the excess is small, one bundle containing at least $q+1$ goods creates a near-unit seed payment, and Lemma~\ref{lem:positive-small-propagation} transmits this requirement through the size-$q$ bundles to almost all agents. In both cases, the lower bound matches the algorithm's leading payment.

\section{Exact divisibility}
\label{sec:monetary-exact-divisibility}

The following theorem characterizes the minimum subsidy when $m_n$ is divisible by $n$.

\begin{theorem}\label{thm:monetary-exact-divisibility}
    Fix an integer $q\ge 0$ and suppose that $m_n=qn$. Then
    \[
        \befmon=
        \begin{cases}
            0,                       & q=0,\\
            \Theta_p(\log n),        & q=1,\\
            0 \quad \text{w.h.p.},  & q\ge 2.
        \end{cases}
    \]
\end{theorem}

The case $q=0$ is immediate, and the case $q\ge2$ follows from the zero-payment EF existence result of Manurangsi and Suksompong~\cite{manurangsi2020envy}. We therefore focus on the square case $q=1$ and assume $m_n=n$ throughout the remainder of this section. We first compute the minimum payment among singleton allocations and show that it equals the unrestricted optimum w.h.p. We then establish the logarithmic bounds in Theorem~\ref{thm:square-logarithmic-scale}.

\subsection{Singleton reduction and computation}

We first consider the special case of \emph{singleton allocations}, in which every agent receives exactly one good. Equivalently, a singleton allocation is a perfect matching between the agents and the goods. Let $B_{\mathrm{sing}}$ denote the minimum total payment among all EF singleton allocations. Throughout the rest of this section, after computing a maximum-weight perfect matching $\sigma:N\to M$, we relabel the goods without loss of generality so that $\sigma(i)=i$ for all $i\in N$, retaining the notation $v_i(j)$ for the relabeled valuations. Algorithm~\ref{alg:square-singleton-compensation} finds an EF singleton allocation with minimum total payment.

\begin{algorithm}[tbp]
\caption{EF allocation and compensation via shortest paths}
\label{alg:square-singleton-compensation}
\begin{algorithmic}[1]
    \Require Valuation matrix $(v_i(j))_{i\in N,j\in M}$, where
    $|N|=|M|=n$
    \Ensure An EF singleton outcome $(A,X)$ with $X\ge0$ and minimum total payment
    \State Compute a maximum-weight perfect matching $\sigma:N\to M$ under
    the valuations $v_i(j)$
    \State Set $A_i\gets\{\sigma(i)\}$ for every $i\in N$
    \State Relabel good $\sigma(i)$ as good $i$ for every $i\in N$, retaining $v_i(j)$ for the relabeled valuations
    \State Set $\ell(i,j)\gets v_i(i)-v_i(j)$ for every
    $i,j\in N$
    \State Compute all-pairs shortest-path distances $d(i,j)$ under the edge
    lengths $\ell$
    \State Set $x_i\gets-\min_{j\in N}d(i,j)$ for every $i\in N$
    \State \Return $A$ and $X$
\end{algorithmic}
\end{algorithm}

The next lemma establishes the correctness and running time of Algorithm~\ref{alg:square-singleton-compensation}.

\begin{lemma}\label{lem:square-singleton-computation}
    The maximum-weight perfect matching $\sigma$ is unique a.s. Let $X=(x_i)_{i\in N}$ be the coordinatewise minimal solution to
    \begin{equation}\label{eq:square-difference-constraints}
        x_i\ge 0,
        \qquad
        x_i-x_j
        \ge v_i(j)-v_i(i)
        \quad\text{for all }i,j\in N.
    \end{equation}
    Then
    \begin{equation}\label{eq:square-singleton-budget}
        B_{\mathrm{sing}}=\sum_{i\in N}x_i.
    \end{equation}
    Moreover, in the complete directed graph on $N$ with edge weight
    \(
        w(i,j):=v_i(j)-v_i(i),
    \)
    we have
    \begin{equation}\label{eq:square-longest-path}
        x_i
        =\max\{\text{weight of a simple directed path starting at $i$}\},
    \end{equation}
    where the empty path has weight zero. On this probability-one event, Algorithm~\ref{alg:square-singleton-compensation} returns an EF singleton allocation with minimum total payment in $O(n^3)$ time.
\end{lemma}

\paragraph{Proof sketch.}
Summing the EF inequalities over any reassignment cancels the payments, so every EF singleton allocation must maximize welfare; absolute continuity makes this matching unique a.s. Its optimality rules out positive-weight cycles in the envy graph. Maximum path weights therefore give the coordinatewise least nonnegative payments, and negating the edge weights reduces their computation to all-pairs shortest paths.

Lemma~\ref{lem:square-singleton-computation} solves the restricted singleton problem. Since $m_n=n$, every nonsingleton allocation contains both an empty bundle and a bundle with at least two goods. The next lemma shows that these allocations are more expensive than optimal singleton allocations w.h.p.

\begin{lemma}\label{lem:square-singleton-reduction}
    Suppose that $m_n=n$. Then
    \begin{equation}\label{eq:square-singleton-reduction}
        \Pr\{\befmon=B_{\mathrm{sing}}\}\longrightarrow 1.
    \end{equation}
\end{lemma}

\paragraph{Proof sketch.}
Apply Lemmas~\ref{lem:positive-large-coverage} and~\ref{lem:positive-small-propagation} with $q=1$ and $r_n=0$. Many empty bundles force payments through many disjoint pairs of goods; with few empty bundles, one pair creates a near-unit payment that propagates to almost all singleton owners. On their common high-probability event, these estimates make every nonsingleton EF outcome cost at least $n-o(n)$. In contrast, Lemmas~\ref{lem:balanced-compensation} and~\ref{lem:positive-dual-oscillation} provide a singleton outcome with $o(n)$ total payment w.h.p. This separation establishes the reduction without using the logarithmic estimates below.

\begin{remark}
For $m_n=n$, Algorithm~\ref{alg:square-singleton-compensation} refines Algorithm~\ref{alg:balanced-compensation} by replacing its dual-price payments with the minimal nonnegative EF payments supporting the same maximum-weight matching.
\end{remark}

\subsection{The logarithmic compensation scale}

By Lemmas~\ref{lem:square-singleton-computation} and \ref{lem:square-singleton-reduction}, it remains to estimate $B_{\mathrm{sing}}=\sum_{i\in N}x_i$. For a realized payment vector $X$, pair good $i$ with payment $x_i$ and take a fresh valuation vector $V=(V_1,\ldots,V_n)$, independent of the original instance, with i.i.d. coordinates drawn from $\mathcal D$. Define the choice share $p_i(X):=\Pr_V\{V_i+x_i=\max_{j\in N}(V_j+x_j)\}$. The probability is over $V$ with $X$ fixed, so $p(X)$ is a deterministic transform of $X$, random across original instances.

Minimality gives $\min_{i\in N}x_i=0$, since \eqref{eq:square-difference-constraints} involves only payment differences. Set $R:=\{i\in N:x_i=0\}$, $K:=|R|$, and $p_{\min}:=\min_{i\in N}p_i(X)$. The indices in $R$ are called \emph{roots}; they attain the smallest choice share.

Choice shares are useful because their Jacobian cancels a determinant in the payment density. Minimality gives each nonroot a nontrivial tight outgoing EF constraint, and uniqueness of the optimal matching rules out tight cycles, yielding a forest toward $R$ almost surely. For each fixed root set, summing over these forests via the matrix--tree theorem gives a Laplacian determinant. The same determinant is the reduced Jacobian determinant of the choice map, so changing variables leaves the face density $n!p_{\min}^K$, independent of $\mathcal D$. Appendix~\ref{app:square-logarithmic-proofs} gives the details and establishes $K=\Theta_p(\sqrt n)$. Conditioning on $R$ then yields the Gamma and exponential representation in Lemma~\ref{lem:square-choice-share-representation}.

\begin{lemma}
    \label{lem:square-choice-share-representation}
    \emergencystretch=2em
    Given the realized root set $R$, and hence $K=|R|$, let $Z_0\sim\operatorname{Gamma}(K+1,1)$ and $E_i\sim\operatorname{Exp}(1)$, $i\notin R$, be mutually independent, where the Gamma distribution uses the shape--rate parametrization, and set $S:=Z_0+\sum_{i\notin R}E_i$. Then the conditional distribution of the choice-share vector is that of the vector given by
    \[
        p_i=
        \begin{cases}
            \dfrac{Z_0}{nS}, & i\in R,\\[4pt]
            \dfrac{Z_0}{nS}+\dfrac{E_i}{S}, & i\notin R.
        \end{cases}
    \]
\end{lemma}

Using the representation in Lemma~\ref{lem:square-choice-share-representation} and $K=\Theta_p(\sqrt n)$, Lemma~\ref{lem:square-root-scale} in Appendix~\ref{app:square-logarithmic-proofs} establishes $p_{\min}=\Theta_p(n^{-3/2})$ and $\max_{i\in N}p_i=O_p(\log n/n)$. In particular,
\begin{equation}
    T_n:=-\log(np_{\min})=\frac12\log n+O_p(1).     \label{eq:square-choice-share-scale}
\end{equation}
This statistic measures the logarithmic gap between the minimum choice share and the average share $1/n$. Lemma~\ref{lem:square-share-budget-comparison} relates this gap to $B_{\mathrm{sing}}$ using the density bounds $\alpha\le f\le\beta$.

\begin{lemma}\label{lem:square-share-budget-comparison}
    For $T_n=-\log(np_{\min})$ defined above,
    \[
        \frac{T_n}{\beta}-O_p(1)
        \le B_{\mathrm{sing}}
        \le \frac{T_n}{\alpha}+O(1).
    \]
\end{lemma}

We can now combine the choice-share estimates with the singleton reduction to obtain logarithmic bounds for the unrestricted optimum $\befmon$.

\begin{theorem}\label{thm:square-logarithmic-scale}
    Suppose that $m_n=n$. For every fixed $\epsilon>0$,
    \[
        \Pr\!\left\{
            \left(\frac{1}{2\beta}-\epsilon\right)\log n
            \le \befmon
            \le
            \left(\frac{1}{2\alpha}+\epsilon\right)\log n
        \right\}\longrightarrow 1.
    \]
    In particular, $\befmon=\Theta_p(\log n)$.
\end{theorem}

\paragraph{Proof sketch.}
Substitute $T_n=\tfrac12\log n+O_p(1)$ from \eqref{eq:square-choice-share-scale} into Lemma~\ref{lem:square-share-budget-comparison}. The resulting tight error terms are smaller than $\epsilon\log n$ w.h.p., and Lemma~\ref{lem:square-singleton-reduction} transfers the bounds from $B_{\mathrm{sing}}$ to $\befmon$. This proves the $q=1$ case of Theorem~\ref{thm:monetary-exact-divisibility} and completes the classification.

\begin{remark}
For the uniform distribution on $[0,1]$, Lemmas~\ref{lem:square-singleton-reduction}, \ref{lem:square-share-budget-comparison}, and~\ref{lem:square-root-scale} further yield $\befmon=\tfrac12\log n+O_p(1)$.
\end{remark}

\section{Conclusion}

We studied the minimum total monetary subsidy needed for exact envy-freeness under our i.i.d. additive valuation model. For $m_n=qn+r_n$ with fixed $q$, the optimum is $n-r_n+o_p(n)$ in the nonzero-remainder cases analyzed. The square case $m_n=n$ has logarithmic compensation. Together with the known zero-subsidy guarantee for $m_n=qn$ and $q\ge2$, these results show how the number of goods and its divisibility by the number of agents shape typical compensation requirements.

A natural future direction is to replace money by heterogeneous divisible goods. For comparison with a monetary budget $B$, assume that each agent values the entire divisible endowment at most $B$, while valuations of its parts may differ. This setting has been studied in the worst case by Li et al.~\cite{li2024ratio}, who give EF guarantees up to a fraction of one good depending on each agent's indivisibility ratio, namely the fraction of her total value attributable to indivisible goods. Under random valuations of the indivisible goods, it remains to determine how preference diversity affects the minimum $B$ needed for exact envy-freeness with high probability and when this quantity is asymptotically smaller than the required monetary budget.

\begin{credits}
\subsubsection{\ackname} GPT-5.6 Sol was used to assist with mathematical analysis, language editing, and manuscript formatting.
The author has verified all mathematical details and takes full responsibility for the content of this paper.

% \subsubsection{\discintname}
% It is now necessary to declare any competing interests or to specifically
% state that the authors have no competing interests. Please place the
% statement with a bold run-in heading in small font size beneath the
% (optional) acknowledgments\footnote{If EquinOCS, our proceedings submission
% system, is used, then the disclaimer can be provided directly in the system.},
% for example: The authors have no competing interests to declare that are
% relevant to the content of this article. Or: Author A has received research
% grants from Company W. Author B has received a speaker honorarium from
% Company X and owns stock in Company Y. Author C is a member of committee Z.
\end{credits}
\bibliographystyle{splncs04}
\bibliography{asymptotic}




% LNCS Homepage, \url{http://www.springer.com/lncs}, last accessed 2023/10/25
% \end{thebibliography}

\clearpage
\appendix

% These environments restate a numbered result without advancing any theorem
% counter. The number is inherited from the result's label in the main text.
\makeatletter
\newenvironment{restatedtheorem}[2][]{%
  \def\restated@note{#1}%
  \begin{trivlist}
  \item[\hskip\labelsep{\bfseries Theorem~\ref{#2}%
    \ifx\restated@note\@empty.\else\ (\restated@note).\fi}]%
  \itshape
}{\end{trivlist}}
\newenvironment{restatedlemma}[2][]{%
  \def\restated@note{#1}%
  \begin{trivlist}
  \item[\hskip\labelsep{\bfseries Lemma~\ref{#2}%
    \ifx\restated@note\@empty.\else\ (\restated@note).\fi}]%
  \itshape
}{\end{trivlist}}
\makeatother

\section{Missing Proofs of Section~\ref{sec:monetary-nonzero-remainders}}
\label{app:sec3-proofs}

This appendix restates the results from Section~\ref{sec:monetary-nonzero-remainders} whose proofs were deferred and gives their complete proofs. We first recall the main theorem; its two cases are proved after the corresponding auxiliary lemmas.

\begin{restatedtheorem}{thm:monetary-nonzero-remainders}
Fix an integer $q\ge0$ and suppose that $m_n=qn+r_n$, where $1\le r_n<n$. If $q=0$, assume additionally that $r_n/n\to\lambda<1$. Then
\[
    \befmon=
    \begin{cases}
        n-r_n+O_p(1), & q=0,\\
        n-r_n+o_p(n), & q\ge1.
    \end{cases}
\]
Moreover, Algorithm~\ref{alg:balanced-compensation} returns an EF outcome for every valuation profile in $O(n^3)$ time, and under the random model its total payment obeys the corresponding estimate above.
\end{restatedtheorem}

\begin{restatedlemma}
    {lem:balanced-compensation}
For every valuation profile, Algorithm~\ref{alg:balanced-compensation} returns an EF outcome, and its total payment satisfies
\begin{align*}
    \sum_{i\in N}x_i
    =n\max_{j\in N}\psi_j-\sum_{j\in N}\psi_j.
\end{align*}
The algorithm runs in $O(nm_n+n^3)$ arithmetic operations, and hence in $O(n^3)$ operations when $q$ is fixed.
\end{restatedlemma}

\begin{proof}
The bipartite matching polytope is integral. Since the graph is balanced and complete and all weights are nonnegative, an optimal integral matching can be extended, without decreasing its weight, to a perfect matching that remains optimal. Thus the matching step is well defined. Complementary slackness gives $\phi_i+\psi_j=v_i(P_j)$ whenever $i$ is matched to $P_j$. Consequently, for every $i,h\in N$,
\[
    v_i(A_h)\le\phi_i+\pi_h,
    \qquad
    v_i(A_i)=\phi_i+\pi_i,
\]
and hence $v_i(A_h)-v_i(A_i)\le\pi_h-\pi_i=x_i-x_h$. This is precisely the EF inequality. The payments are nonnegative by the definition of $H$. Because every bundle is matched and $\psi_j\ge0$, we have $H=\max_{i\in N}\pi_i=\max_{j\in N}\psi_j$. Summing the payments now gives~\eqref{eq:balanced-payment}.

The bundle values can be computed in $O(nm_n)$ operations. A primal--dual assignment algorithm computes the matching and a complementary optimal dual solution in $O(n^3)$ operations. Since fixed $q$ implies $m_n=O(n)$, the final claim follows.
\end{proof}

\subsection{The sparse case: $q=0$}

\begin{restatedlemma}
    {lem:sparse-uniform-deficiency}
Suppose that $0<m_n<n$ and $m_n/n\to\lambda<1$, and define
\[
    \Delta_n:=
    \max_{\substack{U\subseteq N\\|U|=n-m_n}}
    \sum_{g\in M}\left(1-\max_{i\in U}v_i(g)\right).
\]
Then $\Delta_n=O_p(1)$.
\end{restatedlemma}

\begin{proof}
For any fixed $(n-m_n)$-subset $U\subseteq N$ and $0\le t\le1$, PDF-boundedness gives
\[
    \Pr\left(1-\max_{i\in U}v_i(g)>t\right)
    \le (1-\alpha t)^{n-m_n}
    \le e^{-\alpha(n-m_n)t}.
\]
It follows that
\[
    \mathbb E\left[
        \exp\left(\frac{\alpha(n-m_n)}{2}
        \left(1-\max_{i\in U}v_i(g)\right)\right)
    \right]\le2.
\]
For a fixed $U$, these deficiencies are independent across goods. Hence an exponential Markov bound followed by a union bound over all possible $U$ gives
\begin{align*}
    \Pr(\Delta_n>z)
    \le \binom{n}{n-m_n}2^{m_n}
        e^{-\alpha(n-m_n)z/2}.
\end{align*}
Since $m_n/n\to\lambda<1$ and $(n-m_n)/n\to1-\lambda>0$, the right-hand side tends to zero for every sufficiently large fixed $z$. Therefore $\Delta_n=O_p(1)$.
\end{proof}

\begin{restatedlemma}
    {lem:sparse-critical-row}
Suppose that $0<m_n<n$ and $m_n/n\to\lambda<1$. For every agent $i$, define $\eta_i:=1-\max_{g\in M}v_i(g)$, and let $\eta_{(m_n+1)}$ be the $(m_n+1)$-st smallest of these values. Then
\[
    (n-m_n)\eta_{(m_n+1)}=O_p(1).
\]
\end{restatedlemma}

\begin{proof}
For fixed $a>0$, the count
\[
    Z_a:=\left|\left\{i\in N:
        \eta_i\le \frac{a}{n-m_n}\right\}\right|
\]
is binomial because the rows of valuations are independent. For all sufficiently large $n$, the density lower bound implies
\begin{align*}
    \Pr\left(\eta_i\le\frac{a}{n-m_n}\right)
    &\ge 1-\left(1-\frac{\alpha a}{n-m_n}\right)^{m_n}\\
    &\ge 1-\exp\left(-\frac{\alpha a m_n}{n-m_n}\right).
\end{align*}
The event $(n-m_n)\eta_{(m_n+1)}>a$ is precisely $Z_a\le m_n$. If $\lambda>0$, choose $a$ such that
\[
    1-\exp\left(-\frac{\alpha a\lambda}{1-\lambda}\right)>\lambda.
\]
The expectation of $Z_a$ is then separated from $m_n$ by a linear margin, so a Chernoff bound gives $\Pr(Z_a\le m_n)\to0$.

If $\lambda=0$, then $m_n=o(n-m_n)$ and $n-m_n\sim n$. For fixed $a$ and all sufficiently large $n$,
\[
    1-\exp\left(-\frac{\alpha a m_n}{n-m_n}\right)
    \ge \frac{\alpha a m_n}{2(n-m_n)},
\]
and hence $\mathbb E[Z_a]\ge\alpha a m_n/2$. If $a\ge4/\alpha$, then $m_n\le\mathbb E[Z_a]/2$, and a Chernoff bound gives
\[
    \Pr(Z_a\le m_n)
    \le \exp\left(-\frac{\mathbb E[Z_a]}{8}\right)
    \le e^{-\alpha a/16},
\]
where the last inequality uses $m_n\ge1$. This probability can be made arbitrarily small by choosing $a$ large, proving the claim.
\end{proof}

\begin{proof}[$q=0$ case of
Theorem~\ref{thm:monetary-nonzero-remainders}]
\emph{Upper bound.} Run Algorithm~\ref{alg:balanced-compensation} and denote the agents matched to empty bundles by $U$. Then $|U|=n-m_n$. If $i\in U$ is matched to the empty bundle $P_j$, complementary slackness gives $\phi_i+\psi_j=v_i(P_j)=0$. Since both potentials are nonnegative, $\phi_i=\psi_j=0$. Write $\psi_g$ for the potential of the singleton bundle containing $g$. Every such potential is at most one: if $\psi_g>1$, reducing it to one preserves dual feasibility because $v_i(g)\le1$ and $\phi_i\ge0$, contradicting optimality. The empty-bundle potentials are zero and every singleton bundle is matched, so \eqref{eq:balanced-payment} becomes
\begin{align}
    \sum_{i\in N}x_i
    =n\max_{g\in M}\psi_g-\sum_{g\in M}\psi_g
    \le n-\sum_{g\in M}\psi_g.                             \label{eq:sparse-budget}
\end{align}

For $i\in U$, dual feasibility and $\phi_i=0$ give $\psi_g\ge v_i(g)$. Hence \eqref{eq:sparse-budget} and Lemma~\ref{lem:sparse-uniform-deficiency} yield
\begin{align*}
    \sum_{i\in N}x_i
    &\le n-m_n+\sum_{g\in M}\left(1-\max_{i\in U}v_i(g)\right)\\
    &\le n-m_n+\Delta_n
     =n-m_n+O_p(1).
\end{align*}
Therefore $\befmon\le n-m_n+O_p(1)$.

\emph{Lower bound.} Use the notation of Lemma~\ref{lem:sparse-critical-row}. Any allocation has at least $n-m_n$ agents with empty bundles. All such agents must receive the same payment: two agents with identical empty bundles would otherwise make the lower-paid agent envious. Moreover, at least one of the $m_n+1$ agents with the smallest $\eta_i$ must receive an empty bundle. That agent values some allocated good at least $1-\eta_{(m_n+1)}$; since all payments are nonnegative, EF requires the common empty-bundle payment to be at least this value. Consequently, Lemma~\ref{lem:sparse-critical-row} gives
\begin{align*}
    \befmon
    &\ge(n-m_n)\bigl(1-\eta_{(m_n+1)}\bigr)\\
    &=n-m_n-O_p(1).
\end{align*}
Combining the two bounds proves the $q=0$ assertion. Since the lower bound applies to every EF outcome, it also applies to the output of Algorithm~\ref{alg:balanced-compensation}, establishing the stated payment guarantee for that algorithm.
\end{proof}

\subsection{The case $q\ge1$}

We assume $q\ge1$ throughout this subsection. The next three lemmas also apply when $r_n=0$, as needed in Section~\ref{sec:monetary-exact-divisibility}.

\begin{restatedlemma}
    {lem:positive-dual-oscillation}
Fix a deterministic partition $P_1,\ldots,P_n$ in which every bundle has size $q$ or $q+1$. For any maximum-weight perfect matching between agents and these bundles, write $i_j$ for the agent matched to $P_j$, and let $(\phi,\psi)$ be any complementary optimal solution to the corresponding assignment dual~\eqref{eq:balanced-dual}. Define $\theta_j:=\psi_j-|P_j|$. There exists a constant $C>0$, depending only on $q$ and $\alpha$, such that, w.h.p.,
\[
    \max_j\theta_j-\min_j\theta_j
    \le C n^{-1/(q+1)}(\log n)^{1+1/(q+1)}=o(1).
\]
This event holds simultaneously for every choice of the maximum-weight matching and complementary optimal dual solution.
\end{restatedlemma}

\begin{proof}
Choose a sufficiently large constant $\gamma>0$, depending only on $q$ and $\alpha$, and set
\[
    t_n=\gamma\left(\frac{\log n}{n}\right)^{1/(q+1)}.
\]
Call an agent--bundle pair $(i,j)$ \emph{high} if
\[
    v_i(g)\ge 1-\frac{t_n}{q+1}
    \qquad\text{for every $g\in P_j$}.
\]
For all sufficiently large $n$, PDF-boundedness gives
\[
    \Pr\bigl((i,j)\text{ is high}\bigr)
    \ge
    \left(\frac{\alpha t_n}{q+1}\right)^{q+1}
    =c_\gamma\frac{\log n}{n},
\]
where $c_\gamma>0$ can be made arbitrarily large by increasing $\gamma$. Since the bundles are disjoint and all singleton valuations are independent, the high-edge indicators are mutually independent.

If $\gamma$ is sufficiently large, then w.h.p. every set of at most $n/4$ vertices on either side of the high-edge bipartite graph has at least twice as many neighbors, and every agent set and bundle set of size at least $n/4$ have an edge between them. Indeed, the probability that either property fails is at most
\[
    2\sum_{s=1}^{\lfloor n/4\rfloor}
        \binom ns\binom n{2s}
        \exp\left(-c_\gamma\frac{s(n-2s)\log n}{n}\right)
    +2^{2n}\exp\left(-\frac{c_\gamma n\log n}{16}\right)
    =o(1).
\]
These two properties hold simultaneously for every bijection between agents and bundles. Contract the pairs $(i_j,P_j)$ of any such bijection and place an arc $j\to h$ whenever $(i_j,h)$ is high. Forward and backward neighborhoods double until they contain at least $n/4$ vertices, after which the second property connects them. The resulting directed graph therefore has diameter $O(\log n)$.

Dual feasibility and complementary slackness give, along every arc $j\to h$,
\[
    v_{i_j}(P_h)-v_{i_j}(P_j)\le \psi_h-\psi_j.
\]
Since $v_{i_j}(P_h)\ge |P_h|-t_n$ and $v_{i_j}(P_j)\le |P_j|$, it follows that
\[
    \theta_h\ge \theta_j-t_n.
\]
Propagating this inequality along a directed path from a maximizer of $\theta$ to a minimizer yields
\[
    \max_j\theta_j-\min_j\theta_j
    =O(t_n\log n)
    =O\left(n^{-1/(q+1)}(\log n)^{1+1/(q+1)}\right).
\]
Because the expansion event holds simultaneously for every bijection, this bound applies to every maximum-weight matching and every complementary optimal dual solution.
\end{proof}

\begin{restatedlemma}
    {lem:positive-large-coverage}
Let $h_n:=\lceil n^{2/3}\rceil$. There exists a constant $\gamma_1>0$, depending only on $q$ and $\alpha$, such that, for $\eta_n:=\gamma_1n^{-1/(3(q+1))}$, the following holds w.h.p. For every family $\mathcal T$ of at least $\lceil h_n/(q+1)\rceil$ pairwise disjoint $(q+1)$-subsets of goods, all but at most $h_n$ agents value every good in at least one member of $\mathcal T$ at least $1-\eta_n$.
\end{restatedlemma}

\begin{proof}
It suffices to consider subfamilies of exactly $L_n:=\lceil h_n/(q+1)\rceil$ tuples. For a fixed such family, a given agent fails to value every good of any one tuple at least $1-\eta_n$ with probability at most
\[
    \left(1-(\alpha\eta_n)^{q+1}\right)^{L_n}
    \le \exp\left(-\Theta(n^{1/3})\right).
\]
By independence across agents, the probability that $h_n$ prescribed agents all fail is $\exp(-\Theta(n))$. The number of choices of these agents and of the retained tuple family is at most
\[
    \binom n{h_n}m_n^{(q+1)L_n}
    =\exp\bigl(O(n^{2/3}\log n)\bigr)
    =\exp(o(n)).
\]
A union bound proves the claim simultaneously for every such family, including families selected after the valuations are observed.
\end{proof}

\begin{restatedlemma}
    {lem:positive-small-propagation}
Set $h_n:=\lceil n^{2/3}\rceil$. There exist positive constants $\gamma_2$, $\gamma_3$, and $\kappa_0$, depending only on $q$ and $\alpha$. With
\[
    \eta_n':=\gamma_2n^{-1/(3(q+1))},\qquad
    \zeta_n:=\gamma_3n^{-1/(6q)},\qquad
    \ell_n:=\lfloor\kappa_0n^{5/6}\rfloor,
\]
the following properties hold simultaneously w.h.p.:
\begin{enumerate}
    \item every $(q+1)$-subset of goods has at least $4h_n$ agents who value each of its goods at least $1-\eta_n'$;
    \item every $q$-subset of goods has at least $\ell_n$ agents who value each of its goods at least $1-\zeta_n/q$;
    \item for every family of at least $\lceil\ell_n/2\rceil$ pairwise disjoint $q$-subsets, all but at most $\lceil n^{1/3}\rceil$ agents value every good in at least one member of the family at least $1-\zeta_n/q$.
\end{enumerate}
\end{restatedlemma}

\begin{proof}
For a fixed $(q+1)$-subset, the number of agents who value each of its goods at least $1-\eta_n'$ stochastically dominates a binomial random variable with success probability $(\alpha\eta_n')^{q+1}=\Theta(n^{-1/3})$. Choosing $\gamma_2$ sufficiently large, a Chernoff bound makes the probability that this number is below $4h_n$ exponentially small in $n^{2/3}$. A union bound over the $O(n^{q+1})$ possible subsets proves the first property.

For a fixed $q$-subset, the corresponding success probability is at least
\[
    \left(\frac{\alpha\zeta_n}{q}\right)^q
    =\Theta(n^{-1/6}).
\]
A Chernoff bound and a union bound over the $O(n^q)$ possible subsets prove the second property after choosing $\kappa_0$ sufficiently small.

For the third property, retain exactly $\lceil\ell_n/2\rceil$ members of any larger family. For a fixed retained family, a given agent fails to value every good in any one member at least $1-\zeta_n/q$ with probability at most
\[
    \exp\left(-\Theta(n^{-1/6}\ell_n)\right)
    =\exp\left(-\Theta(n^{2/3})\right).
\]
Thus $\lceil n^{1/3}\rceil$ prescribed agents all fail with probability $\exp(-\Theta(n))$. The number of choices of these agents and of the retained family is at most
\[
    \binom n{\lceil n^{1/3}\rceil}
        m_n^{q\lceil\ell_n/2\rceil}
    =\exp\bigl(O(n^{5/6}\log n)\bigr)
    =\exp(o(n)).
\]
Another union bound proves the third property uniformly over all such families.
\end{proof}

\begin{proof}[$q\ge1$ case of
Theorem~\ref{thm:monetary-nonzero-remainders}]
\emph{Upper bound.} Run Algorithm~\ref{alg:balanced-compensation} and define $\theta_j:=\psi_j-|P_j|$. By Lemma~\ref{lem:positive-dual-oscillation}, $\max_j\theta_j-\min_j\theta_j=o(1)$ w.h.p. This is the step that converts the dual construction into a sharp estimate of its total payment.

Since $r_n>0$, at least one predetermined bundle has size $q+1$. On the event $\max_j\theta_j-\min_j\theta_j<1$, every maximizer $j^\star$ of $\psi_j$ is such a bundle: its potential exceeds that of every size-$q$ bundle. As $\sum_j|P_j|=m_n=qn+r_n$, \eqref{eq:balanced-payment} becomes
\[
    \sum_{i\in N}x_i
    =n-r_n+n\theta_{j^\star}-\sum_{j\in N}\theta_j.
\]
Consequently,
\begin{align}
    \left|\sum_{i\in N}x_i-(n-r_n)\right|
    \le n\left(\max_j\theta_j-\min_j\theta_j\right)=o(n),
    \label{eq:positive-upper}
\end{align}
which proves both the claimed payment estimate for the algorithm and $\befmon\le n-r_n+o_p(n)$.

\emph{Lower bound.} Work on the intersection of the high-probability events in Lemmas~\ref{lem:positive-large-coverage} and~\ref{lem:positive-small-propagation}. Both events are uniform over the families of goods to which they apply. They may therefore be used for an arbitrary EF allocation $A$ chosen after all valuations are observed. Let $X$ be its payment vector, write $s_i=|A_i|$, and define
\[
    D=\sum_{i\in N}(q-s_i)_+,
    \qquad
    E=\sum_{i\in N}(s_i-q)_+.
\]
Here $(z)_+=\max\{z,0\}$. Since $\sum_i s_i=qn+r_n$,
\begin{align}
    E=D+r_n.                                           \label{eq:positive-balance}
\end{align}
At most $E$ agents have more than $q$ goods, while at most $D\le E$ agents have fewer than $q$ goods. Thus at most $2E$ agents have bundle size different from $q$.

Set $h_n=\lceil n^{2/3}\rceil$. First suppose that $E\ge h_n$. From every bundle larger than $q$, extract as many disjoint $(q+1)$-tuples as possible. The total number extracted is at least $E/(q+1)$ because, for every integer $s>q$,
\[
    \left\lfloor\frac{s}{q+1}\right\rfloor
    \ge\frac{s-q}{q+1}.
\]
Lemma~\ref{lem:positive-large-coverage} thus ensures that all but at most $h_n$ agents have a high tuple in this extracted family. Its uniformity is essential here because both the allocation and hence the tuple family depend on the realized valuations.

For every nonexceptional agent with $s_i\le q$, EF comparison with her high tuple and nonnegativity of the other agent's payment give
\[
    x_i\ge q+1-s_i-(q+1)\eta_n.
\]
Furthermore, \eqref{eq:positive-balance} implies
\[
    \sum_{s_i\le q}(q+1-s_i)
    =|\{i\in N:s_i\le q\}|+D
    \ge n-E+D
    =n-r_n.
\]
After removing the exceptional agents, we obtain
\begin{align}
    \sum_{i\in N}x_i
    &\ge n-r_n-(q+1)h_n-(q+1)n\eta_n \nonumber\\
    &=n-r_n-o(n).                                  \label{eq:positive-large-lower}
\end{align}

It remains to consider $E<h_n$. In this case at least $n-2h_n$ agents receive exactly $q$ goods, and \eqref{eq:positive-balance} gives $r_n\le E=o(n)$. Since $r_n>0$, some allocated bundle contains a $(q+1)$-tuple. By the first part of Lemma~\ref{lem:positive-small-propagation}, at least $4h_n$ agents value every good in this tuple at least $1-\eta_n'$. Excluding the fewer than $2h_n$ agents whose bundles do not have size $q$, one such agent $i_0$ owns exactly $q$ goods. Comparing her bundle with the bundle containing the high tuple gives the seed bound
\begin{align}
    x_{i_0}\ge 1-(q+1)\eta_n'.                       \label{eq:positive-seed}
\end{align}

The last two parts of the same lemma propagate this seed bound. Applied to $A_{i_0}$, the second part supplies $\ell_n$ agents who value each of its goods at least $1-\zeta_n/q$. Since $h_n=o(\ell_n)$, at least $\lceil\ell_n/2\rceil$ of them are exact-$q$ owners for all sufficiently large $n$. For each such owner $j$, EF gives
\[
    x_j\ge x_{i_0}-\zeta_n.
\]
Their bundles form a family of disjoint $q$-subsets. The third part of the lemma and one more EF comparison therefore imply
\[
    x_i\ge x_{i_0}-2\zeta_n
\]
for all but $\lceil n^{1/3}\rceil$ of the remaining exact-$q$ owners. Together with \eqref{eq:positive-seed}, this yields
\begin{align}
    \sum_{i\in N}x_i
    &\ge (n-2h_n-\lceil n^{1/3}\rceil)
        \bigl(1-(q+1)\eta_n'-2\zeta_n\bigr) \nonumber\\
    &=n-o(n)
     =n-r_n-o(n).                                    \label{eq:positive-small-lower}
\end{align}

Because the two valuation events are uniform, these lower bounds hold simultaneously for every EF allocation and payment vector. Hence \eqref{eq:positive-large-lower} and \eqref{eq:positive-small-lower} give $\befmon\ge n-r_n-o_p(n)$. Combining this with \eqref{eq:positive-upper} proves the $q\ge1$ assertion and completes the proof of Theorem~\ref{thm:monetary-nonzero-remainders}.
\end{proof}

\section{Missing Proofs of Section~\ref{sec:monetary-exact-divisibility}}
\label{app:sec4-proofs}

\begingroup
\emergencystretch=2em

This appendix gives the complete proofs for the square case in Section~\ref{sec:monetary-exact-divisibility}. We first establish the singleton computation and reduction, then analyze the choice shares to prove the logarithmic bounds.

We assume $m_n=n$ throughout. Let $B_{\mathrm{sing}}$ denote the minimum total payment among EF allocations in which every agent receives exactly one good.

\subsection{Singleton reduction and computation}

As in the main text, after choosing a maximum-weight perfect matching $\sigma:N\to M$, we relabel the goods so that $\sigma(i)=i$ for all $i\in N$. We retain the notation $v_i(j)$ for the relabeled valuations throughout this appendix.

\begin{restatedlemma}{lem:square-singleton-computation}
    The maximum-weight perfect matching $\sigma$ is unique a.s. Let $X=(x_i)_{i\in N}$ be the coordinatewise minimal solution to
    \begin{equation*}
        x_i\ge 0,
        \qquad
        x_i-x_j
        \ge v_i(j)-v_i(i)
        \quad\text{for all }i,j\in N.
    \end{equation*}
    Then
    \begin{equation*}
        B_{\mathrm{sing}}=\sum_{i\in N}x_i.
    \end{equation*}
    Moreover, in the complete directed graph on $N$ with edge weight
    \(
        w(i,j):=v_i(j)-v_i(i),
    \)
    we have
    \begin{equation*}
        x_i
        =\max\{\text{weight of a simple directed path starting at $i$}\},
    \end{equation*}
    where the empty path has weight zero. On this probability-one event, Algorithm~\ref{alg:square-singleton-compensation} returns an EF singleton allocation with minimum total payment in $O(n^3)$ time.
\end{restatedlemma}

\begin{proof}
We first work with the original good labels. Consider any singleton allocation described by a perfect matching $\tau$ and any nonnegative payment vector $Y=(y_i)_{i\in N}$ that makes it EF. For an arbitrary perfect matching $\rho$, let $\kappa(i):=\tau^{-1}(\rho(i))$ be the agent who receives good $\rho(i)$ under $\tau$. The EF inequalities give
\[
    v_i(\tau(i))+y_i
    \ge v_i(\rho(i))+y_{\kappa(i)}
    \qquad (i\in N).
\]
Because $\kappa$ is a permutation, the payments cancel after summation. Consequently, $\tau$ maximizes welfare. Any two distinct perfect matchings have different total welfare a.s.: after conditioning on all but one entry on which they differ, equality has probability zero by absolute continuity. There are only finitely many matchings, so the welfare-maximizing matching $\sigma$ is unique a.s.

We now use the relabeling convention above. It remains to find the least payments that support $\sigma$. The EF inequality comparing agent $i$ with the bundle of agent $j$ is precisely
\[
    x_i-x_j\ge
    v_i(j)-v_i(i)=w(i,j).
\]
For every directed cycle, the sum of its edge weights is the change in welfare obtained by cyclically reassigning the corresponding goods. The optimality of $\sigma$ therefore implies that every directed cycle has nonpositive weight.

Define $x_i$ by the right-hand side of \eqref{eq:square-longest-path}. The maximum is well defined because a walk can be made simple by deleting nonpositive cycles, and the empty path shows that $x_i\ge0$. Prepend the edge $i\to j$ to a maximizing path from $j$. Deleting any resulting repeated segment again cannot decrease its weight, so $x_i\ge w(i,j)+x_j$. Thus $X$ satisfies all constraints in \eqref{eq:square-difference-constraints}. Conversely, if $Y$ is any nonnegative feasible vector, telescoping its constraints along a path from $i$ to $j$ shows that $y_i-y_j$ is at least the weight of that path. Since $y_j\ge0$, $y_i$ dominates every path weight and hence $y_i\ge x_i$. Therefore $X$ is the coordinatewise minimal feasible vector and $B_{\mathrm{sing}}=\sum_i x_i$.

Finally, the lengths used by Algorithm~\ref{alg:square-singleton-compensation} satisfy $\ell(i,j)=-w(i,j)$. They have no negative directed cycle, and hence
\[
    -\min_{j\in N}d(i,j)
    =\max\{\text{weight of a directed path starting at $i$}\}=x_i.
\]
A maximum-weight perfect matching can be computed in $O(n^3)$ time. All-pairs shortest paths can also be computed in $O(n^3)$ time, for example using the Floyd--Warshall algorithm. All remaining steps take $O(n^2)$ time, which proves the algorithmic claim.
\end{proof}

\begin{restatedlemma}{lem:square-singleton-reduction}
    Suppose that $m_n=n$. Then
    \begin{equation*}
        \Pr\{\befmon=B_{\mathrm{sing}}\}\longrightarrow 1.
    \end{equation*}
\end{restatedlemma}

\begin{proof}
We compare singleton and nonsingleton allocations on uniform high-probability events. Apply Lemmas~\ref{lem:positive-large-coverage} and~\ref{lem:positive-small-propagation} with $q=1$ and $r_n=0$, and write $h_n=\lceil n^{2/3}\rceil$. Their conclusions hold simultaneously for every relevant family of goods, including families determined after the valuations are observed.

Fix any nonsingleton allocation $A$ and any payments $X$ that make it EF. Let $e$ be the number of empty bundles, $b$ the number of bundles containing at least two goods, and $s=n-e-b$ the number of singleton bundles. Since there are exactly $n$ goods,
\[
    \sum_{|A_j|\ge2}(|A_j|-1)=e,
    \qquad b\le e,
\]
and $e\ge1$.

Suppose first that $e\ge h_n$. From the bundles of size at least two we can extract at least $e/2$ pairwise disjoint pairs of goods, because $\lfloor a/2\rfloor\ge(a-1)/2$ for every integer $a\ge2$. Lemma~\ref{lem:positive-large-coverage} implies that all but at most $h_n$ agents value both goods in at least one extracted pair at least $1-\eta_n$. For each such agent $i$ with $|A_i|\le1$, EF and the nonnegativity of the pair holder's payment give $x_i\ge2-|A_i|-2\eta_n$. Moreover,
\[
    \sum_{|A_i|\le1}(2-|A_i|)=2e+s=n+e-b\ge n.
\]
Removing at most $h_n$ exceptional agents therefore yields, uniformly over all such allocations,
\[
    \sum_{i\in N}x_i\ge n-2h_n-2n\eta_n=n-o(n).
\]

Now suppose that $e<h_n$. Fewer than $2h_n$ agents have nonsingleton bundles, so at least $n-2h_n$ agents own one good. Choose a pair in any bundle of size at least two. By the first part of Lemma~\ref{lem:positive-small-propagation}, at least $4h_n$ agents value both goods in this pair at least $1-\eta_n'$. One of them, say $i_0$, must own a singleton, and EF gives
\[
    x_{i_0}\ge1-2\eta_n'.
\]
The second part of the same lemma supplies $\ell_n$ agents who value the good owned by $i_0$ at least $1-\zeta_n$. Since $h_n=o(\ell_n)$, at least $\lceil\ell_n/2\rceil$ of them own singletons, and each such owner $j$ satisfies $x_j\ge x_{i_0}-\zeta_n$. Their goods are pairwise distinct. The third part of the lemma, followed by one further EF comparison, now gives $x_i\ge x_{i_0}-2\zeta_n$ for all but $\lceil n^{1/3}\rceil$ of the singleton owners. Consequently,
\[
    \sum_{i\in N}x_i
    \ge(n-2h_n-\lceil n^{1/3}\rceil)
       (1-2\eta_n'-2\zeta_n)
    =n-o(n).
\]
Thus, on one common high-probability event, every nonsingleton EF outcome requires total payment $n-o(n)$.

It remains to exhibit a cheaper singleton outcome. Apply Lemma~\ref{lem:positive-dual-oscillation} with $q=1$ to the deterministic partition into singleton goods. Choose a maximum-weight perfect matching and a complementary optimal dual solution $(\phi,\psi)$, and let agent $i_j$ be matched to good $j$. The payments $y_{i_j}=\max_h\psi_h-\psi_j$ are EF by the argument of Lemma~\ref{lem:balanced-compensation}. Moreover, with $\theta_j=\psi_j-1$,
\[
    \sum_{i\in N}y_i
    \le n\left(\max_j\theta_j-\min_j\theta_j\right)
    =O\!\left(\sqrt n(\log n)^{3/2}\right)=o(n)
\]
w.h.p. Hence $B_{\mathrm{sing}}=o_p(n)$. Intersecting the two high-probability events shows that, for all sufficiently large $n$, every nonsingleton EF outcome is strictly more expensive than some singleton EF outcome. Therefore $\Pr\{\befmon=B_{\mathrm{sing}}\}\to1$.
\end{proof}

\subsection{The logarithmic compensation scale}
\label{app:square-logarithmic-proofs}

We recall the notation needed for the change of variables. Let $X$ be the coordinatewise minimal nonnegative EF payment vector supporting $\sigma$. For a fresh valuation vector $V$, independent of the original instance, with i.i.d. coordinates drawn from $\mathcal D$, set
\[
    p_i(x):=\Pr_V\{V_i+x_i=\max_j(V_j+x_j)\}.
\]
Write $R:=\{i:x_i=0\}$, $K:=|R|$, $p_{\min}:=\min_i p_i(X)$, and $p_{\max}:=\max_i p_i(X)$. Minimality gives $\min_i x_i=0$, so $R$ is nonempty. Let $F$ be the CDF of $\mathcal D$, extended by $F(t)=0$ for $t\le0$ and $F(t)=1$ for $t\ge1$, and extend $f$ by zero outside $[0,1]$.

\begin{lemma}
\label{lem:square-root-face-change}
Fix a nonempty set $R\subseteq N$, put $S=N\setminus R$ and $k=|R|$, and consider the relative interior of the root face
\[
    \mathcal X_R:=\{x\in\mathbb R_{\ge0}^n:
        x_r=0\ (r\in R),\ x_i>0\ (i\in S)\}.
\]
For $x\in\mathcal X_R$ and any fixed $r\in R$, define
\begin{align*}
    a(x)&:=\int_{\mathbb R}
        f(t)\prod_{\ell\ne r}F(t-x_\ell)\,dt,\\
    W_{ij}(x)&:=\int_{\mathbb R}
        f(t-x_i)f(t-x_j)
        \prod_{\ell\notin\{i,j\}}F(t-x_\ell)\,dt
        \qquad(i\ne j).
\end{align*}
The definition of $a(x)$ does not depend on the selected root because all root coordinates are zero. Let $L(W)$ be the weighted Laplacian, with $L(W)_{ii}=\sum_{j\ne i}W_{ij}$ and $L(W)_{ij}=-W_{ij}$ for $i\ne j$. The minimal payment vector $X$ has density
\begin{equation}
    \rho_R(x)=n!\,a(x)^k\det L(W)_{S,S}                 \label{eq:root-face-density}
\end{equation}
with respect to $dx_S$ on this face. Moreover, the choice map satisfies
\begin{equation}
    p_r(x)=a(x)\quad(r\in R),
    \qquad
    \det\!\left(\frac{\partial p_S}{\partial x_S}\right)
       =\det L(W)_{S,S}.                               \label{eq:choice-jacobian}
\end{equation}
On $\mathcal X_R\cap\{\max_i x_i<1\}$, equivalently on the positive-density part of this face, it maps bijectively to the part of the open probability simplex whose minimum coordinates are exactly those in $R$. Consequently,
\begin{equation}
    \rho_R(x)\,dx_S=n!\,p_{\min}(x)^k\,dp_S.           \label{eq:root-share-change}
\end{equation}
\end{lemma}

\begin{proof}
We first compute the density on a root face. Outside a null set, each full-dimensional cell fixes for every nonroot $i$ a unique nontrivial tight constraint
\[
    x_i-x_{\phi(i)}=v_i(\phi(i))-v_i(i),
\]
and these tight edges $i\to\phi(i)$ form a directed forest toward $R$. Indeed, a nonroot coordinate with no tight outgoing constraint could be decreased, contradicting minimality. A directed tight cycle would have zero weight and produce a second welfare-maximizing matching, contradicting the a.s. uniqueness proved in Lemma~\ref{lem:square-singleton-computation}. On a cell with a fixed forest, write $d_i=v_i(i)$. The tight entries are $v_i(\phi(i))=d_i+x_i-x_{\phi(i)}$. The change of variables from $(d_1,\ldots,d_n,x_S)$ to the diagonal and tight-parent entries has absolute Jacobian one: after the identity block in the $d_i$ variables is removed, the remaining matrix is the reduced incidence matrix of a forest and is triangular when vertices are ordered toward the roots.

All unused entries in row $i$ satisfy $v_i(j)\le d_i+x_i-x_j$. With $t=d_i+x_i$, integrating a root row gives $a(x)$, whereas a nonroot row whose parent is $j$ contributes $W_{ij}(x)$. Thus a fixed forest contributes $a(x)^k\prod_{i\in S}W_{i\phi(i)}(x)$. Summing over all forests directed toward $R$ and applying the all-minors matrix--tree theorem gives $a(x)^k\det L(W)_{S,S}$ on a specified matching cell. The $n!$ possible optimal matchings form disjoint cells up to null boundaries and have equal density by exchangeability of the columns. This proves \eqref{eq:root-face-density}.

For an independent valuation row, conditioning on its winning adjusted value gives
\begin{equation}
    p_i(x)=\int_{\mathbb R}
        f(t-x_i)\prod_{\ell\ne i}F(t-x_\ell)\,dt.      \label{eq:choice-integral}
\end{equation}
In particular, all roots have share $a(x)$. A perturbation of size $\varepsilon$ can change the winner only if two adjusted values are within $O(\varepsilon)$ of each other. Since their differences have bounded densities, the choice map is locally Lipschitz. Rademacher's theorem and Fubini therefore justify the following derivatives at almost every $x$:
\[
    \frac{\partial p_i}{\partial x_j}=-W_{ij}\quad(i\ne j),
    \qquad
    \frac{\partial p_i}{\partial x_i}=\sum_{j\ne i}W_{ij};
\]
the second identity also follows by differentiating $\sum_i p_i=1$. Therefore $Dp=L(W)$ and the reduced Jacobian is the one in \eqref{eq:choice-jacobian}.

For completeness, we verify the range of this map. On the hyperplane $H=\{x:\sum_i x_i=0\}$, set $\Phi(x)=\mathbb E[\max_i(V_i+x_i)]$. The maximizing coordinate is unique a.s., so $\nabla\Phi=p$. On $\Omega:=\{x\in H:\max_i x_i-\min_i x_i<1\}$, every $W_{ij}$ is positive and
\[
    u^\top L(W)u=\sum_{i<j}W_{ij}(u_i-u_j)^2>0
    \qquad(0\ne u\in\mathbf1^\perp).
\]
Since $p$ is locally Lipschitz, the first derivative of the restriction of $\Phi$ to any nonconstant line segment in $\Omega$ is absolutely continuous, and its a.e. second derivative is the positive quadratic form displayed above. Hence $\Phi$ is strictly convex there. For any point $z$ in the open probability simplex, $\Phi(x)-z\cdot x$ is coercive on $H$. Indeed, if $M=\max_i x_i$, then $\Phi(x)-z\cdot x\ge\sum_i z_i(M-x_i)
\ge n(\min_i z_i)M$, which diverges with $\lVert x\rVert$ on $H$.
It therefore has a minimizer. Its first-order condition says that $p(x)-z$ is orthogonal to $H$; since both vectors sum to one, this means $p(x)=z$. Every such minimizer lies in $\Omega$, since a coordinate lying at least one below the maximum would have zero choice share. Conversely, every coordinate has positive share on $\Omega$: it can win when its value is near one and all other values are near zero. Strict convexity on $\Omega$ makes the minimizer unique. This proves that $p:\Omega\to\{z:z_i>0,\ \sum_i z_i=1\}$ is a bijection.

The bijection preserves strict coordinate order. Indeed, if $\tau$ swaps coordinates $i$ and $j$, exchangeability gives $p(\tau x)=\tau p(x)$, while strict convexity gives $(p(x)-p(\tau x))\cdot(x-\tau x)>0$ whenever $x_i\ne x_j$. Hence $x_i>x_j$ if and only if $p_i(x)>p_j(x)$. Translating to the gauge $\min_i x_i=0$ shows that the root face maps to the claimed simplex region.

If $\max_i x_i\ge1$ on the nonnegative gauge, a root has zero share and the factor $a(x)^k$ in the density vanishes. Faces with additional roots are lower-dimensional boundaries, so only the positive-density part matters. The choice map is injective there, and hence the area formula applies with multiplicity one. It cancels the two determinants and gives \eqref{eq:root-share-change}.
\end{proof}

\begin{lemma}\label{lem:square-root-count}
    For every $1\le k\le n$,
    \[
        \Pr(K=k)=\frac{k(n-1)!}{n^k(n-k)!}.
    \]
\end{lemma}

\begin{proof}
Fix a $k$-element set $R$ and put $S=N\setminus R$. By Lemma~\ref{lem:square-root-face-change}, its probability is obtained by integrating $n!a^k\,dp_S$ over the simplex region in which the coordinates in $R$ share the common minimum $a$. Set
\[
    y_i:=p_i-a\quad(i\in S),
    \qquad Z:=na=1-\sum_{i\in S}y_i.
\]
Then $y_i>0$, $\sum_{i\in S}y_i<1$, and
\[
    \left|\det\frac{\partial p_S}{\partial y_S}\right|
    =\det\!\left(I-\frac1n\mathbf1\mathbf1^\top\right)=\frac{k}{n}.
\]
The probability of this specified root set is therefore
\begin{align*}
    \frac{n!k}{n^{k+1}}
    \int_{\substack{y_i>0\\\sum_{i\in S}y_i<1}}
        \left(1-\sum_{i\in S}y_i\right)^kdy_S
    =\frac{k\,k!}{n^{k+1}},
\end{align*}
where the last equality is the standard Dirichlet integral. This formula also covers $k=n$, with a zero-dimensional integral. Finally, there are $\binom nk$ possible root sets, so
\[
    \Pr(K=k)=\binom nk\frac{k\,k!}{n^{k+1}}
             =\frac{k(n-1)!}{n^k(n-k)!}.
\]
\end{proof}

\begin{restatedlemma}{lem:square-choice-share-representation}
    
    Given the realized root set $R$, and hence $K=|R|$, let $Z_0\sim\operatorname{Gamma}(K+1,1)$ and $E_i\sim\operatorname{Exp}(1)$, $i\notin R$, be mutually independent, where the Gamma distribution uses the shape--rate parametrization, and set $S:=Z_0+\sum_{i\notin R}E_i$. Then the conditional distribution of the choice-share vector is that of the vector given by
    \[
        p_i=
        \begin{cases}
            \dfrac{Z_0}{nS}, & i\in R,\\[4pt]
            \dfrac{Z_0}{nS}+\dfrac{E_i}{S}, & i\notin R.
        \end{cases}
    \]
\end{restatedlemma}

\begin{proof}
Retain the notation in the proof of Lemma~\ref{lem:square-root-count}. Before integration, the change-of-variables density conditional on the specified root set $R$ is proportional to
\[
    Z^K\,\mathbf1\!\left\{y_i>0,
        \ \sum_{i\notin R}y_i+Z=1\right\}.
\]
After normalization, this is exactly
\[
    \bigl((y_i)_{i\notin R},Z\bigr)
       \sim\operatorname{Dirichlet}(1^{n-K},K+1).
\]
A Dirichlet vector is obtained by normalizing independent Gamma variables. Thus we may write $y_i=E_i/S$ and $Z=Z_0/S$, where the variables and $S$ are as in the statement. Since $a=Z/n$, every root has share $p_i=a$, while every nonroot has share $p_i=a+y_i$. Substitution gives the claimed representation.
\end{proof}

Set $T_n:=-\log(np_{\min})$. The next two lemmas determine the scale of this statistic and compare it with $B_{\mathrm{sing}}$.

\begin{lemma}\label{lem:square-root-scale}
    We have $K=\Theta_p(\sqrt n)$ and $p_{\min}=\Theta_p(n^{-3/2})$. Moreover, $\max_{i\in N}p_i=O_p(\log n/n)$ and
    \[
        T_n=\frac12\log n+O_p(1).
    \]
\end{lemma}

\begin{proof}
The mass function in Lemma~\ref{lem:square-root-count} is the difference of successive terms in the tail identity
\[
    \Pr(K\ge k)=\prod_{j=1}^{k-1}\left(1-\frac{j}{n}\right),
    \qquad 1\le k\le n.
\]
For $k=\lceil t\sqrt n\rceil$ with fixed $t>0$, a Taylor expansion gives
\[
    \log\Pr(K\ge k)
    =\sum_{j=1}^{k-1}\log\left(1-\frac{j}{n}\right)
    =-\frac{k(k-1)}{2n}+o(1)\longrightarrow-\frac{t^2}{2}.
\]
Hence $K/\sqrt n$ converges in distribution to a random variable $R_0$ with $\Pr(R_0>t)=e^{-t^2/2}$. In particular, $K=\Theta_p(\sqrt n)$ and $K\to\infty$ in probability.

Conditional on the realized root set $R$, the representation in Lemma~\ref{lem:square-choice-share-representation} and Gamma concentration give $Z_0/K\to_p1$, uniformly over root sets with $K\to\infty$. The conditional law of $S$ is $\operatorname{Gamma}(n+1,1)$ and does not depend on $R$, so $S/n\to_p1$. It follows that
\[
    n^{3/2}p_{\min}
    =\frac{Z_0/K}{S/n}\frac{K}{\sqrt n}
    \ \Longrightarrow\ R_0,
\]
which proves $p_{\min}=\Theta_p(n^{-3/2})$. On the event $K<n$,
\[
    p_{\max}\le p_{\min}
       +\frac{\max_{i\notin R}E_i}{S}.
\]
A union bound gives $\max_{i\notin R}E_i=O_p(\log n)$, while $S=\Theta_p(n)$. Since $\Pr(K=n)=n!/n^n=o(1)$, this proves the claimed upper bound on $p_{\max}$.

Finally, $np_{\min}=Z_0/S$, and hence
\[
    T_n=\log\frac{S}{Z_0}
       =\log\frac{n}{K}+o_p(1)
       =\frac12\log n-\log\frac{K}{\sqrt n}+o_p(1)
       =\frac12\log n+O_p(1),
\]
where the last step uses that $R_0$ is finite and strictly positive a.s.
\end{proof}

\begin{restatedlemma}{lem:square-share-budget-comparison}
    For $T_n=-\log(np_{\min})$ defined above,
    \[
        \frac{T_n}{\beta}-O_p(1)
        \le B_{\mathrm{sing}}
        \le \frac{T_n}{\alpha}+O(1).
    \]
\end{restatedlemma}

\begin{proof}
Write $B=B_{\mathrm{sing}}=\sum_i x_i$ and choose a root $r\in R$. By conditioning on the fresh value $V_r$ and substituting $d=1-V_r$, its choice share is
\begin{equation}
    p_{\min}
    =\int_0^1 f(1-d)\prod_{j\ne r}F(1-d-x_j)\,dd.      \label{eq:root-share-integral}
\end{equation}
For every $y\ge0$, the lower density bound gives $F(1-y)\le e^{-\alpha y}$, with the extension of $F$ covering $y>1$. Since $x_r=0$ and $f\le\beta$,
\[
    p_{\min}
    \le \beta e^{-\alpha B}
       \int_0^1e^{-\alpha(n-1)d}\,dd
    \le \frac{\beta}{\alpha(n-1)}e^{-\alpha B}.
\]
Taking logarithms yields $B\le T_n/\alpha+O(1)$. In particular, Lemma~\ref{lem:square-root-scale} implies $B=O_p(\log n)$.

We next show that no individual payment is macroscopically large. Put $Q=\max_i x_i$, choose an index $j^\star$ attaining this maximum, set $H=\{i:x_i>Q/2\}$, and write $p_{\max}=\max_i p_i$. If
\[
    V_{j^\star}>1-\min\{Q/2,1\},
\]
then the good chosen by the fresh agent lies in $H$: the adjusted value of $j^\star$ is strictly larger than the maximum possible adjusted value outside $H$. This event has probability at least $\alpha\min\{Q/2,1\}$. Since $|H|\le2B/Q$ whenever $Q>0$, we obtain
\begin{equation}
    \alpha\min\{Q/2,1\}
    \le\sum_{i\in H}p_i
    \le |H|p_{\max}
    \le\frac{2Bp_{\max}}{Q}.                           \label{eq:max-payment-control}
\end{equation}
The bounds already proved give $Bp_{\max}=o_p(1)$. Equation \eqref{eq:max-payment-control} first implies $\Pr(Q>2)\to0$; on $Q\le2$ it then gives $Q^2\le4Bp_{\max}/\alpha$. On this event,
\[
    QB\le\frac{2}{\sqrt\alpha}B^{3/2}p_{\max}^{1/2}.
\]
Together with $B=O_p(\log n)$, $p_{\max}=O_p(\log n/n)$, and $\Pr(Q>2)\to0$, this proves $Q=o_p(1)$ and $QB=o_p(1)$.

The upper density bound implies that, for all sufficiently small $y\ge0$,
\[
    \log F(1-y)\ge-\beta y-C_\beta y^2
\]
for a constant $C_\beta$. Restrict \eqref{eq:root-share-integral} to $0\le d\le1/n$. Uniformly on this interval, $d+x_j=o_p(1)$ and
\[
    \sum_{j\ne r}(d+x_j)^2
    \le 2(n-1)d^2+2\sum_jx_j^2
    \le\frac{2}{n}+2QB=o_p(1).
\]
Using also $f\ge\alpha$ in~\eqref{eq:root-share-integral}, we conclude that for a constant $c_{\alpha,\beta}>0$,
\[
    p_{\min}
    \ge\frac{c_{\alpha,\beta}}{n}
       e^{-\beta B-o_p(1)}.
\]
Taking logarithms now gives $B\ge T_n/\beta-O_p(1)$, completing the comparison.
\end{proof}

\begin{restatedtheorem}{thm:square-logarithmic-scale}
    Suppose that $m_n=n$. For every fixed $\epsilon>0$,
    \[
        \Pr\!\left\{
            \left(\frac{1}{2\beta}-\epsilon\right)\log n
            \le \befmon
            \le
            \left(\frac{1}{2\alpha}+\epsilon\right)\log n
        \right\}\longrightarrow 1.
    \]
    In particular, $\befmon=\Theta_p(\log n)$.
\end{restatedtheorem}

\begin{proof}
Lemma~\ref{lem:square-root-scale} and Lemma~\ref{lem:square-share-budget-comparison} give
\[
    \frac{1}{2\beta}\log n-O_p(1)
    \le B_{\mathrm{sing}}
    \le \frac{1}{2\alpha}\log n+O_p(1).
\]
For any fixed $\epsilon>0$, the two $O_p(1)$ terms have magnitude at most $\epsilon\log n$ w.h.p. Lemma~\ref{lem:square-singleton-reduction} gives $\befmon=B_{\mathrm{sing}}$ on another event whose probability tends to one. Intersecting these events proves the displayed probability statement. Taking, for example, any $\epsilon<1/(2\beta)$ also yields $\befmon=\Theta_p(\log n)$.
\end{proof}

\endgroup

\end{document}